\documentclass[10pt]{article}
\usepackage[utf8]{inputenc}
\usepackage{amsmath} 
\usepackage{amssymb}
\usepackage{indentfirst}
\usepackage{mathrsfs}
\usepackage{multirow}
\usepackage{geometry}
\usepackage{booktabs}  
\usepackage{bm}
\usepackage{enumerate}
\usepackage{graphicx}
\usepackage{listings}
\usepackage{float}
\usepackage{graphicx}
\usepackage{caption}
\usepackage{subcaption}
\usepackage{comment}
\usepackage{amsthm}

\usepackage{natbib}
\usepackage{xcolor}

\DeclareMathOperator*{\argmin}{argmin}

\DeclareMathOperator*{\cov}{cov}
\DeclareMathOperator*{\pr}{pr}

\newtheorem{theorem}{Theorem}
\newtheorem{proposition}{Proposition}
\newtheorem{lemma}{Lemma}

\newcommand{\footremember}[2]{%
    \footnote{#2}
    \newcounter{#1}
    \setcounter{#1}{\value{footnote}}%
}

\title{Inference for non-stationary time series quantile regression with inequality constraints}

\author{Yuan Sun$^1$ \and Zhou Zhou$^2$\footremember{alley}{Email: zhou.zhou@utoronto.ca}}
\date{%
    $^1$Lunenfeld-Tanenbaum Research Institute, Sinai Health, Toronto, Canada\\%
    $^2$Department of Statistical Sciences, University of Toronto, Toronto, Canada\\[2ex]%
}

\begin{document}
\maketitle

\begin{abstract}
We consider parameter inference for linear quantile regression with non-stationary predictors and errors, where the regression parameters are subject to inequality constraints. We show that the constrained quantile coefficient estimators are asymptotically equivalent to the metric projections of the unconstrained estimator onto the constrained parameter space. Utilizing a geometry-invariant property of this projection operation, we propose inference procedures - the Wald, likelihood ratio, and rank-based methods - that are consistent regardless of whether the true parameters lie on the boundary of the constrained parameter space. We also illustrate the advantages of considering the inequality constraints in analyses through simulations and an application to an electricity demand dataset. 
\end{abstract}

\section{Introduction}
Quantile regression has become a powerful method for analyzing the distributional relationship between the responses and the predictors since the seminal work by \cite{KB1978}. While a significant amount of work focuses on scenarios with independent observations, quantile regression has been studied under various time series settings. For example, the quantile autoregression model and the quantile autoregressive conditional heteroskedasticity model  were proposed in \cite{KX2006} and \cite{KZ1996}, respectively. \cite{P1991} studied the asymptotics of regression quantiles in an $m$-dependent setting. \cite{KM1994} considered the case where the errors are stationary and long-range dependent Gaussian random variables. 

Non-stationary time series analyses have attracted increasing attention in recent years, as non-stationary behaviors have been observed in temporally ordered data collected from a wide range of practical applications. 
\cite{Z2013} introduced a piecewise locally stationary process, which allows the underlying data-generating mechanism of the series to change abruptly around a finite number of breakpoints and smoothly evolve in between. Due to its ability to capture general forms of non-stationary behavior in both predictors and errors, this piecewise locally stationary framework has been assumed in subsequent time series literature such as \cite{Z2015,WZ2018,RS2019}, among others.
We also adopt the piecewise locally stationary framework in this paper, and we refer the readers to Section 2.2 for the detailed definition and discussion of piecewise locally stationary time series.
See also \citet{DPV2011,KP2015,DRW2019,HHY2019,DP2021,K2022,BR2023} among others for recent developments on non-stationary time series analysis.

Consider the following non-stationary time series quantile regression  at a given quantile level $\tau$:
\begin{equation}\label{eq:model} 
    y_i=x_i^T \beta_0(\tau) + \epsilon_{i,\tau},\quad  i=1,2,\ldots,n,
\end{equation}
where $\{x_i=(1,x_{i2},\ldots,x_{ip})\}_{i=1}^n$ is a $p$ dimensional piecewise locally stationary time series of predictors that always include the intercept, and $\beta_0(\tau)=(\beta_1(\tau),...,\beta_p(\tau))$ is a $p$ dimensional vector of coefficients. We shall write $\beta_0(\tau)$ as $\beta_0$ throughout this paper to simplify the notation. For identifiability, we require the $\tau$th conditional quantile of the piecewise locally stationary error process $\epsilon_{i,\tau}$ given $x_i$ to be 0. 

In this paper, we consider the inference of model \eqref{eq:model} when the regression coefficients $\beta_0$ are subject to inequality constraints. Inequality constraints are sometimes necessary to ensure model validity (e.g., an autoregressive conditional heteroskedasticity model requires all coefficients to be non-negative). There are also scenarios where prior knowledge suggests that certain constraints should be imposed. In demand analysis, it is usually reasonable to assume that the demand for a product decreases as the product's price increases, so the coefficient of price could constrained to be non-positive when regressing demand on price.
As another example, \cite{HN1999} studied the degradation of roof flashing of U.S. army bases and naturally assumed that the percentage of roof flashing in good condition could only decrease over time.

Although inequality constraints are commonly encountered in applications, they are sometimes overlooked in the analyses due to the lack of available methods. However, taking the inequality constraints into account offers at least two advantages. First, when the inequality constraints are not considered, it can be difficult to carry out further analysis when the fitted parameters fail to satisfy the constraints. Second, considering the inequality constraints can restrict the parameters into a smaller space, thereby improving the estimation accuracy and hypothesis testing power.

For quantile regression, \cite{KN2005} proposed an algorithm for parameter computation under inequality constraints. \cite{P2019} studied the asymptotics of the constrained quantile process for independent data. \cite{LZL2020} and \cite{WLY2021} considered $l_1$-penalized quantile regression with inequality constraints. \cite{QY2015} utilized constrained quantile regression to ensure monotonicity in their nonparametric quantile process model. 
As far as we know, no results on (non-stationary) time series quantile regression with inequality constraints are available in the literature. 

In this paper, we aim to develop inference methods for constrained quantile regression where both the predictors and the errors are piecewise locally stationary. With the observation that the constrained quantile estimator of $\beta_0$ can be approximated by a matrix projection of the unconstrained estimator, we derive the limiting distribution of the constrained quantile coefficient estimator. We also consider a likelihood ratio test and rank-based test for parameter inference under our setting and establish their asymptotic properties. 

However, direct inference based on our asymptotic results is challenging because the limiting distributions of the estimated coefficients and test statistics are non-standard and involve 1) the matrix projection operation, which is not continuous when the coefficients are at the boundary; and 2) the conditional density of the errors and the long-run covariance matrix, both of which are both unknown and change over time with possible jumps. To address these issues, we propose a projected multiplier bootstrap procedure to approximate the limiting distributions.

Our bootstrap algorithm utilizes a simple convolution of block sums of the quantile regression gradient vectors with i.i.d. standard normal random variables to consistently approximate the limiting distribution of the unconstrained estimator under complex temporal dynamics. The key in the projected multiplier bootstrap is to notice that the projection direction can be estimated consistently using the Powell sandwich estimator (\cite{Powell1991}) under smoothly and abruptly time-varying data generating mechanisms of the predictors and errors. The limiting distribution of the constrained estimator can then be approximated by projecting the convolution term from the multiplier bootstrap onto this estimated direction. The geometry-invariant property of the projection operation ensures the consistency of the projected multiplier bootstrap procedure, regardless of whether $\beta_0$ lies on the boundary of the constraints.


The remainder of this paper is organized as follows. In Section 2, we formally introduce the problem settings and review the piecewise locally stationary framework. Section 3 shows our main results. More specifically, we study the asymptotic properties of the constrained quantile estimator, propose the likelihood ratio test and the rank-based test in Section 3.1, and introduce the projected multiplier bootstrap algorithm in Section 3.2. Simulation studies and a real data example are given in Sections 4 and 5, respectively. Section 6 presents the regularity conditions. The proofs of our theoretical results are given in the Appendix.

\section{Preliminaries}
\subsection{Settings}

In model \eqref{eq:model}, assume that $\beta_0$ satisfies the inequality constraints $C\beta_0\geq c$, where $C$ is a $q \times p$ full rank matrix with $1\leq q\leq p$ and $c$ is a $q$ dimensional vector. By transformation of variables, the constraints can be simplified into 
\begin{equation}
\label{constraints}
    \beta_0 \in Q,\quad Q=\{(\beta_1,\ldots,\beta_p)\,|\,\beta_j \geq 0,j=1,\ldots,q\}.
\end{equation}

Let $\tilde{\beta}_n$ be the estimated coefficients when the inequality constraints (\ref{constraints}) are ignored, \cite{KB1978} showed that 
\begin{equation}
\label{beta_tilde}
    \tilde{\beta}_n = \argmin \limits_{\beta\in{\mathbb{R}^p}}\sum_{i=1}^n \rho_{\tau}(y_i-x_i^T \beta),
\end{equation}
where $\rho_{\tau}(x)=x(\tau-I(x<0))$ is the so-called check function. 
Then $\hat{\beta}_n$, the estimated coefficient under the inequality constraints, can be naturally estimated by
\begin{equation}
\label{beta_hat}
    \hat{\beta}_n = \argmin\limits_{\beta\in Q}\sum_{i=1}^n \rho_{\tau}(y_i-x_i^T \beta).
\end{equation}
Solving (\ref{beta_hat}) is a quadratic programming problem and can be tackled with the algorithms proposed in \cite{KN2005}. However, the asymptotic behavior of $\hat{\beta}_n$ with piecewise locally stationary predictors and errors is unclear and will be investigated in this paper.

\subsection{Piecewise Locally Stationary Time Series Models}

We adopt the class of piecewise locally stationary processes in \cite{Z2013} to model the predictors and errors.

We call $\{\epsilon_{i,\tau}\}_{i=1}^n$ a piecewise locally stationary process generated by filtration $\mathcal{F}_i$ and $\mathcal{G}_i$ with $R$ break points if there exist constants $0=b_0<b_1<\ldots<b_R<b_{R+1}=1$ and nonlinear filters $D_0,...,D_R$, such that
\begin{equation*}
    \epsilon_{i,\tau}=D_r(t_i,\mathcal{F}_i,\mathcal{G}_i),\, b_r<t_i\leq b_{r+1},
\end{equation*}
where $t_i=i/n$, $\mathcal{F}_i=\{\ldots,\eta_0,\eta_1,\ldots,\eta_i\}$, $\mathcal{G}_i=\{\ldots,\zeta_0,\zeta_1,\ldots,\zeta_i\}$, and $\{\eta_i\}_{i=-\infty}^{\infty}$ and $\{\zeta_i\}_{i=-\infty}^{\infty}$ are independent i.i.d random variables.
 Without loss of generality, we assume $\{x_i\}_{i=1}^n$ shares same break points as $\{\epsilon_{i,\tau}\}_{i=1}^n$, and let
\begin{equation*}
    x_i=H_r(t_i,\mathcal{F}_{i-1},\mathcal{G}_i),\, b_r<t_i\leq b_{r+1},
\end{equation*}
where $H_0,\ldots,H_R$ are non-linear filters. 

The piecewise locally stationary process can capture a broad of class non-stationary behavior in practice because it allows the underlying data generating mechanism to evolve smoothly between breakpoints (provided that the filters are smooth in $t$) while undergoing abrupt changes at these breakpoints. Note that we include the filtration $\mathcal{F}_{i-1}$ and $\mathcal{F}_i$ into $x_i$ and $\epsilon_{i,\tau}$, respectively, to accommodate possible auto-regressive behavior in the predictors; that is, $x_i$ may contain lagged values of the response. Other information that may influence both the predictors and the errors is captured in $\mathcal{G}_i$.
Examples of piecewise locally stationary processes under the current formulation can be found in \cite{WZ2018}.

To study the asymptotic property of a piecewise locally stationary process, we need to define a measure of its temporal dependence structure. Intuitively, the dependence of a process can be evaluated by replacing the inputs ($\eta_i$ and $\zeta_i$) $k$ steps earlier with corresponding i.i.d. copies and comparing the change in the output ($x_i$ and $\epsilon_{i,\tau}$). A larger change in the output indicates stronger dependence. Let $\|\cdot\|_v=(E(| \cdot | )^v)^{1/v}$ denote the $\mathcal{L}_v$ norm and assume $\max_{1\leq i\leq n}\|\epsilon_{i,\tau}\|_v<\infty$ for some $v>1$, we define the $k$th dependence measure for $\{\epsilon_{i,\tau}\}_{i=1}^n$ in $\mathcal{L}_v$ norm as 
\begin{equation*}
    \Delta_v(D,k)=\max_{0\leq r\leq R}\sup_{b_r<t<b_{r+1}}\|D_r(t,\mathcal{F}_k,\mathcal{G}_k)-D_r(t,\mathcal{F}_k^*,\mathcal{G}_k^*)\|_v,
\end{equation*}
where $\mathcal{F}_k^*=\{\eta_k,\eta_{k-1},\ldots,\eta_0^*,\eta_{-1},\ldots\}$, $\eta_0^*$ is independent of $\{\eta_i\}_{i=-\infty}^{\infty}$ and is identically distributed as $\eta_0$, and the filtration $\mathcal{G}_k^*$ is defined in the same way.
The $k$th dependence measure for $\{x_i\}_{i=1}^n$ is defined similarly as 
\begin{equation*}
    \Delta_v(H,k)=\max_{0\leq r\leq R}\sup_{b_r<t<b_{r+1}}\|H_r(t,\mathcal{F}_{k-1},\mathcal{G}_k)-H_r(t,\mathcal{F}_{k-1}^*,\mathcal{G}_k^*)\|_v.
\end{equation*}

\section{Methodology and Its Theoretical Properties}

\subsection{Test Statistics}
Suppose that we are interested in testing:
\begin{equation}
\label{test}
    H_0:\beta_0^{(A)}=0,\, \beta_0\in Q \quad vs \quad H_{\alpha}:\beta_0^{(A)}\neq0,\,\beta_0\in Q
\end{equation}
where $A=\{a_1,\ldots,a_q\}\subset\{1,\ldots,p\}$ with $a_1<\ldots<a_q$ be a set of index and $x^{(A)}=(x_{a_1},\ldots,x_{a_q})$ for a $p$ dimensional vector $x=(x_1,\ldots,x_p)$. We will consider the likelihood ratio test and the rank-based test as these two types of test are widely used for quantile coefficients inference of independent data without inequality constraints (\cite{KM1999}).

Let $A^c$ denotes the complement of set $A$, define the likelihood ratio test as
\begin{equation}
\label{LR}
    T_n^{LR}=\sum_{i=1}^n\big(\rho_{\tau}(y_i-(x_i^{(A^c)})^{\top}\hat{\beta}_n^{A^c})-\rho_{\tau}(y_i-x_i^{\top}\hat{\beta}_n)\big),
\end{equation}
where $\hat{\beta}_n^{A^c}$ is the estimate of $\beta$ under the restricted model that only includes $x_i^{(A^c)}$ as covariates. 
The test statistic $T_n^{LR}$ is the likelihood ratio test defined in Chapter 3 of \cite{K2005} without normalization. It compares the empirical loss under the restricted model and the full model, and a large value of $T_n^{LR}$ is in favor of the alternative hypothesis. 

Let $\psi(u)=\tau-I(u<0)$ be the left derivative function of $\rho_{\tau}(\cdot)$. Define the rank-based test as 
\begin{equation}
\label{RB}
    T_n^{RB}=(S_{1,n}-S_{0,n})^{\top}D_n^{-1}(S_{1,n}-S_{0,n}),
\end{equation}
where $S_{1,n}=\sum_i \psi_{\tau}(y_i-x_i^{\top}\hat{\beta}_n)x_i^{(A)}$, $S_{0,n}=\sum_i \psi_{\tau}(y_i-(x_i^{(A^c)})^{\top}\hat{\beta}_n^{A^c})x_i^{(A)}$ and $D_n=\sum_i (x_i^{(A)})^{\top}(x_i^{(A)})$. 
Note that the rank-based test in \cite{K2005} is constructed with the regression rankscores, which are the solutions to the dual problem of Equation (\ref{beta_tilde}). Because the regression rankscores for observation $i$ at $\tau$ could be approximated with $\tau-\psi_{\tau}(y_i-x_i^{\top}\hat{\beta}_n)$, we construct our rank-based test with the $\psi_{\tau}$ function directly.
Unlike the rank-based test without inequality constraints, $T_n^{RB}$ requires fitting both the restricted model and the full model because $S_{1,n}$ may not be 0 with inequality constraints imposed. 
Also note that while $D_n$ (times a constant) standardizes the rank-based test with no inequality constraints and independent observations, this is not the case in our setting. We still include $D_n$ in our test statistic $T_n^{RB}$ for consistency with other quantile regression rank-based tests.

To study the properties of $\hat{\beta}_n$, $T_n^{LR}$ and $T_n^{RB}$, we need the following lemma.

\begin{lemma}
Under Conditions (C1)-(C4) given in Section 6, 
\begin{equation*} 
\begin{aligned}
    \sup\limits_{|\beta-\beta_0|\leq n^{-1/2}\log n}|\sum_{i=1}^n(\rho_{\tau}(y_i-x_i^{\top}\beta)-\rho_{\tau}&(y_i-x_i^{\top}\beta_0))+(\beta-\beta_0)^{\top}G_n\\
    &-\frac{1}{2}(\beta-\beta_0)^{\top}K_n(\beta-\beta_0)|=o_p(1),
\end{aligned}
\end{equation*}
\begin{sloppypar}
where $G_n=\sum_{i=1}^n x_i\psi_{\tau}(\epsilon_{i,\tau})$ and $K_n=\sum_{i=1}^n E\big(f_r(\frac{i}{n},0\mid\mathcal{F}_{i-1},\mathcal{G}_i)x_i x_i^{\top}\big)$, $f_r(t,x|\mathcal{F}_{k-1},\mathcal{G}_k)=\frac{\partial}{\partial x} \pr(D_r(t,\mathcal{F}_k,\mathcal{G}_k)\leq x\mid\mathcal{F}_{k-1},\mathcal{G}_k)$, for $b_r<t\leq b_{r+1}$. 
\end{sloppypar}
\end{lemma}

Lemma 1 shows that the difference in the check loss function $\sum_{i=1}^n(\rho_{\tau}(y_i-x_i^{\top}\beta)-\rho_{\tau}(y_i-x_i^{\top}\beta_0))$ can be approximated by a quadratic function of $\beta$. This result is well-known when observations are independent (\cite{BRW1992}). We show that it also holds when the predictors and errors are piecewise locally stationary. 

Since Lemma 1 holds for any $|\beta-\beta_0|\leq n^{-1/2}\log n$, it naturally holds for such $\beta$ that meets the inequality constraints. Therefore, the convexity of $\rho_\tau$ implies the consistency of both $\tilde{\beta}_n$ and $\hat{\beta}_n$. The Bahadur representation $(\tilde{\beta}_n-\beta_0)-K_n^{-1} G_n=o_p(n^{-1/2})$ can also be derived from Lemma 1. The asympototic normality of $\tilde{\beta}_n$ can then be established by showing that $K_n/n$ converges to a matrix $\mathcal{M}$ and $n^{-1/2} G_n$ converges to a normal distribution $U$ with mean 0 and covariance $\int_0^1 \Omega(t)\,dt$.
The explicit forms of $\mathcal{M}$ and $\Omega(t)$ are given later.

Define the metric projection onto region $Q$ with respect to a positive definite symmetric matrix $\Sigma$ as
\begin{equation}\label{pro}
    \mathcal{P}_{Q,\Sigma}(\cdot)=\argmin\limits_{\beta\in Q}(\beta-\cdot)^{\top}\Sigma(\beta-\cdot).
\end{equation}
For $x\in \mathbb{R}^p$ and $\beta \in Q$, let
\begin{equation}
\label{theta}
    \Theta_{Q,\Sigma}(\beta,x)=\lim\limits_{n\rightarrow\infty}(\mathcal{P}_{Q,\Sigma}(n\beta+x)-n\beta).
\end{equation}
By Proposition 1 in \cite{Z2015}, if $\Sigma$ is positive definite, $\mathcal{P}_{Q,\Sigma}(a_1\beta+x)-a_1\beta=\mathcal{P}_{Q,\Sigma}(a_2\beta+x)-a_2\beta$ for $a_1$ and $a_2$ large enough, regardless of whether $\beta$ is on the boundary of $Q$. This geometry-invariant property guarantees the existence of the limit in Equation (\ref{theta}).

The metric projection $\mathcal{P}_{Q,\Sigma}(\cdot)$ plays a key role in investigating the asymptotic properties of $\hat{\beta}_n$. 
For least-squares regressions, $\hat{\beta}_n=\mathcal{P}_{Q,\Sigma}(\tilde{\beta}_n)$ with $\Sigma=\sum_{i=1}^n x_i x_i^T/n$ (\cite{Z2015}).
Such a relationship does not hold for quantile regressions, but we shall show that $\hat{\beta}_n$ could be approximated by the projection of $\tilde{\beta}_n$ with respect to a different $\Sigma$.

By Lemma 1 and the Bahadur representation of $\tilde{\beta}_n$, we have
\begin{equation*}
\begin{aligned}
    \sup\limits_{|\beta-\beta_0|\leq n^{-1/2}\log n,\beta\in Q}|\sum_{i=1}^n(&\rho_{\tau}(y_i-x_i^{\top}\beta)-\rho_{\tau}(y_i-x_i^{\top}\beta_0))\\
    &+(\beta-\beta_0)^{\top}K_n(\tilde{\beta}_n-\beta_0)-\frac{1}{2}(\beta-\beta_0)^{\top}K_n(\beta-\beta_0)|=o_p(1).
\end{aligned}
\end{equation*}
Let $\hat{\gamma}_n=\mathcal{P}_{Q,K_n/n}(\tilde{\beta}_n)$, the projection of $\tilde{\beta}_n$ onto $Q$ with respect to the matrix $K_n/n$. We notice that $\hat{\gamma}_n$ is the minimizer of $(\beta-\beta_0)^{\top}K_n(\tilde{\beta}_n-\beta_0)-\frac{1}{2}(\beta-\beta_0)^{\top}K_n(\beta-\beta_0)$. Using the properties of convex functions, we can show that $\hat{\beta}_n$ could by approximated by $\hat{\gamma}_n$. Because
\begin{equation*}
\begin{aligned}
    \surd{n}(\hat{\beta}_n-\beta_0) &\approx \surd{n} \mathcal{P}_{Q,\frac{k_n}{n}}(\tilde{\beta}_n) - \surd{n} \beta_0 \\
    &\approx \mathcal{P}_{Q,\frac{k_n}{n}}(\surd{n}\beta_0 + \surd{n}K_n^{-1}G_n) - \surd{n} \beta_0,
\end{aligned}
\end{equation*}
it is then easy to show that 
\begin{equation}
\label{asym_beta}
    \surd{n}(\hat{\beta}_n-\beta_0)\Rightarrow\Theta_{Q,\mathcal{M}}(\beta_0,\mathcal{M}^{-1}U),
\end{equation}
where $\Rightarrow$ denotes weak convergence.
The limiting distributions of $T_n^{LR}$ and $T_n^{RB}$ are also be derived by Equation (\ref{asym_beta}) and continuous mapping theorem. The above discussion is formally summarized in Theorem 1 below.

\begin{sloppypar}
The following notations are needed for Theorem 1. 
Define the long-term covariance matrix 
    $$\Omega(t)=\sum_{i=-\infty}^{\infty}\cov\{H_r(t,\mathcal{F}_{-1},\mathcal{G}_0)\psi_{\tau}(D_r(t,\mathcal{F}_0,\mathcal{G}_0)),H_r(t,\mathcal{F}_{i-1},\mathcal{G}_i)\psi_{\tau}(D_r(t,\mathcal{F}_i,\mathcal{G}_i))\}$$ for $b_r<t\leq b_{r+1}$.
Write $\mathcal{M}=\int_0^1 M(t)dt$ where $M(t)=E\{f_r(t,0\mid\mathcal{F}_{-1},\mathcal{G}_0) H_r(t,\mathcal{F}_{-1},\mathcal{G}_0)H_r(t,\mathcal{F}_{-1},\mathcal{G}_0)^{\top}\}$ for $b_r<t\leq b_{r+1}$, and $\mathcal{M}_0=\int_0^1 M_0(t)dt$ where $M_0(t)=E\{f_r(t,0\mid\mathcal{F}_{-1},\mathcal{G}_0)H_r^{(A^c)}(t,\mathcal{F}_{-1},\mathcal{G}_0)H_r^{(A^c)}(t,\mathcal{F}_{-1},\mathcal{G}_0)^{\top}\}$.
For the rank-based test, let $\mathcal{M}_0^{RB}=\int_0^1 M_0^{RB}(t)dt$ where $M_0^{RB}(t)=E\{f_r(t,0\mid\mathcal{F}_{-1},\mathcal{G}_0)H_r^{(A)}(t,\mathcal{F}_{-1},\mathcal{G}_0)H_r^{(A^c)}(t,\mathcal{F}_{-1},\mathcal{G}_0)^{\top}\}$, and $\mathcal{M}^{RB}=\int_0^1 M^{RB}(t)dt$ where $M^{RB}(t)=E\{f_r(t,0\mid\mathcal{F}_{-1},\mathcal{G}_0)H_r^{(A)}(t,\mathcal{F}_{-1},\mathcal{G}_0)H_r(t,\mathcal{F}_{-1},\mathcal{G}_0)^{\top}\}$.
Write $\mathcal{D}=\int_0^1 D(t)dt$ where $D(t)=E\{H_r^{(A)}(t,\mathcal{F}_{-1},\mathcal{G}_0)H_r^{(A)}(t,\mathcal{F}_{-1},\mathcal{G}_0)^{\top}\}$.
\end{sloppypar}

\begin{sloppypar}
\begin{theorem}
Under regularity conditions (C1)-(C5), we have

(\romannumeral1) $\surd{n}(\hat{\beta}_n-\beta_0)\Rightarrow\Theta_{Q,\mathcal{M}}(\beta_0,\mathcal{M}^{-1}U)$, where $U$ follows a normal distribution with mean 0 covariance $\int_0^1 \Omega(t)\,dt$;

(\romannumeral2) Under $H_0$, $T_n^{LR}\Rightarrow g_1\{\Theta_{Q,\mathcal{M}_0}(\beta_0^{(A^c)},\mathcal{M}_0^{-1}U^{(A^c)}),\mathcal{M}_0,U^{(A^c)}\}-g_1\{\Theta_{Q,\mathcal{M}}\big(\beta_0,\mathcal{M}^{-1}U),\mathcal{M},U\}$, where $g_1(x,y,z)=\frac{1}{2}x^{\top}yx-z^{\top}x$;

(\romannumeral3) Under $H_0$, $T_n^{RB}\Rightarrow g_2\{\mathcal{M}^{RB}\Theta_{Q,\mathcal{M}}\big(\beta_0,\mathcal{M}^{-1}U),\mathcal{M}^{RB}_0\Theta_{Q,\mathcal{M}_0}(\beta_0^{(A^c)},\mathcal{M}_0^{-1}U^{(A^c)}),\mathcal{D}\}$, where $g_2(x,y,z)=(x-y)^{\top}z^{-1}(x-y)$.
\end{theorem}
\end{sloppypar}

The limiting distributions of $\surd{n}(\hat{\beta}_n-\beta_0)$, $T_n^{LR}$ and $T_n^{RB}$ involves unknown quantities $\Omega(t)$, $f_r(t,0\mid\mathcal{F}_{-1},\mathcal{G}_0)$ and $\beta_0$. Because $\Omega(t)$ and $f_r(t,0\mid\mathcal{F}_{-1},\mathcal{G}_0)$ may change abruptly at breakpoints under the piecewise locally stationary framework, estimating them directly is challenging. Furthermore, we cannot simply replace $\beta_0$ with $\hat{\beta}_n$ because $\Theta_{Q,\Sigma}(\beta,x)$ is not continuous in $\beta$. Additionally, $\hat{\beta}_n$ has a non-zero mass at a point on the boundary of the convex cone $Q$ if the true $\beta_0$ is on the latter boundary. As a result, the asymptotic distribution given in Theorem 1 cannot be applied directly  for inference. Instead, we use the projected multiplier bootstrap as a solution to these challenges.

\subsection{The Projected Multiplier Bootstrap}

We first consider bootstrapping the limiting distribution of $\surd{n}(\hat{\beta}_n-\beta_0)$.

To approximate the behaviour of $U$, let $m$ be a pre-specified block size, and $m^*=n-m+1$. Define
\begin{equation}
\label{Psi1}
    \Psi_m=\sum_{i=1}^{m^*}\frac{1}{\surd{(mm^*)}}(\varpi_{i,m}-\frac{m}{n}\varpi_{1,n})V_i,
\end{equation}
where $\varpi_{i,m}=\sum_{j=i}^{i+m-1}\psi_{\tau}(\hat{\epsilon}_{j,\tau}) x_j$, $\hat{\epsilon}_{i,\tau}=y_i-x_i^{\top}\hat{\beta}_n$ and $\{V_i\}_{i=1}^n$ follow i.i.d standard normal distributions which are independent of $\{\mathcal{F}_i\}_{i=-\infty}^{\infty}$ and $\{\mathcal{G}_i\}_{i=-\infty}^{\infty}$. 
While the block bootstrap is commonly used in time series analysis, as indicated in
\cite{Z2015}, this method is unable to preserve the complex dependence structure of the piecewise locally stationary processes currently assumed. Instead, the multiplier bootstrap, as given in Equations (\ref{Psi1}), offers an alternative. This approach  convolutes the block sums of the gradient vectors from the quantile regression, which locally estimate the covariance matrix $\Omega(t)$, with i.i.d standard normal weights $V_i$.
It is shown in \cite{WZ2018} that the behaviour of $U$ can be approximated by $\Psi_m$.

Let $\phi(\cdot)$ be a kernel density that satisfies $\int \phi(x)dx=0$, $\int \phi(x)x^2dx\leq M$, $\int \phi^2(x)dx\leq M$, and $\int \phi^{\prime 2}(x)dx\leq M$ for some positive constant $M$. Define 
\begin{equation}
 \label{Xi1}   \Xi_{h}=\sum_{i=1}^{n}\frac{\phi(\hat{\epsilon}_{i,\tau}/h)x_ix_i^{\top}}{nh},
\end{equation}
where $h\rightarrow0$ is a bandwidth parameter. This Powell sandwich estimator $\Xi_{h}$ has been widely used to estimate the covariance matrix for quantile regression with independent observations (\cite{Powell1991}). We can show that it can consistently estimate the matrix $\mathcal{M}$ (namely the projection direction) in our settings.

Our remaining task is to approximate the projection operation. Because
$$\surd{n}(\hat{\beta}_n-\beta_0)\approx \mathcal{P}_{Q,\Xi_{h}}(\surd{n}\beta_0+\Xi_{h}^{-1}\Psi_m)-\surd{n}\beta_0,$$
one obvious way is to replace $\beta_0$ with $\hat{\beta}_n$ in the above equation. However, as we mentioned before, this naive replacement is inconsistent because under $H_0$ ($\beta_0$ is at the boundary), $\surd{n}\hat{\beta}_n=\surd{n}\beta_0+O_p(1)$ and the $O_p(1)$ error term is not ignorable.
Alternatively, by the geometry-invariant property of the projection operation, when $n$ is large enough, 
$$\mathcal{P}_{Q,\Xi_{h}}(\surd{n}\beta_0+\Xi_{h}^{-1}\Psi_m)-\surd{n}\beta_0 = \mathcal{P}_{Q,\Xi_{h}}(n^{1/4}\beta_0+\Xi_{h}^{-1}\Psi_m)-n^{1/4}\beta_0.$$ 
Therefore we could instead replace $n^{1/4}\beta_0$ using $n^{1/4}\hat{\beta}_n$ to reduce the multiplication error to 0 asymptotically.

In practice, the block size $m$ in $\Psi_m$ and the bandwidth $h$ in $\Xi_{h}$ should be carefully chosen. The approximation of $U$ will be biased if $m$ is too small to capture the dependence structure, while a larger $m$ will induce a greater variance as the effective number of summands in $\Psi_m$ is smaller. The bandwidth $h$ plays similar role as the bandwidth of a kernel density estimation, and a larger $h$ is associated with larger bias but smaller variance.

We use the minimum volatility method (Chapter 9 of \cite{PRW1999}) for the selection of block size $m$. The idea is that the estimated value should be stable when $m$ is chosen within the reasonable range.
More specifically, let $m_1<\ldots<m_K$ be $K$ candidates of block size. For each candidate $m_k$, we calculate 
$$\hat{\Psi}_{m_k}=\sum_{j=1}^{m_k^*}\frac{(\varpi_{j,m_k}-\frac{m_k}{n}\varpi_{1,n})^{\top}(\varpi_{j,m_k}-\frac{m_k}{n}\varpi_{1,n})}{m_k(n-m_k+1)}.$$
The $m_k$ that minimizes the standard error of $\{\hat{\Psi}_{m_{k+j}}\}_{j=-3}^{j=3}$ is selected.

The bandwidth $h$ is chosen with cross-validation. \cite{RD2019} demonstrated that cross-validation is asympototic optimal for the bandwidth selection in certian local stationary processes. Although our setting differs from that of \cite{RD2019}, we observe that the performance of cross-validation is reasonably good in our simulations.

The detailed steps to construct the confidence interval using the proposed bootstrap method are summarized below. 
\begin{enumerate}[Step 1:]
    \item Select the block size $m$ with the minimum volatility method and the bandwidth $h$ with cross-validation.
    \item Fit the constrained quantile regression with $x_i$ get $\hat{\beta}_n$. Calculate $\Xi_{h}$ as Equation (\ref{Xi1}).
    \item For $i=1,\ldots,m^*$, generate i.i.d standard normal variables $V_i$, and calculate $\Psi_m$ as Equation (\ref{Psi1}).
    Compute $\hat{\Lambda}_n=\mathcal{P}_{Q,\Xi_{h}}(n^{\frac{1}{4}}\hat{\beta}_n+\hat{\Upsilon}_n)-n^\frac{1}{4}\hat{\beta}_n$ where $\hat{\Upsilon}_n=\Xi_{h}^{-1}\Psi_m$.
    \item Repeat Step 3 for $B$ iterations to get $\{\hat{\Lambda}_{n,1},\ldots,\hat{\Lambda}_{n,B}\}$. Let $\hat{q}_{\alpha/2}$ and $\hat{q}_{1-\alpha/2}$ be the $\alpha/2$-th and $(1-\alpha/2)$-th sample percentile of $\{\hat{\Lambda}_{n,1},\ldots,\hat{\Lambda}_{n,B}\}$. The $100(1-\alpha)\%$ confidence interval of $\beta_n$ is given by $(\hat{\beta}_n - \hat{q}_{1-\alpha/2}/\surd{n},\hat{\beta}_n-\hat{q}_{\alpha/2}/\surd{n})$.
 
\end{enumerate}

The limiting distribution of $T_n^{LR}$ and $T_n^{RB}$ under the null hypothesis can be bootstrapped similarly.
Define
\begin{equation}
\label{Xi0}
\Xi_{0,h}=\sum_{i=1}^{n}\frac{\phi(\hat{\epsilon}_{i,\tau}^{A^c}/h)x_i^{(A^c)}x_i^{(A^c)\top}}{nh},   
\end{equation}
\begin{equation}
\label{Xi_RB}
\Xi_{h}^{RB}=\sum_{i=1}^{n}\frac{\phi(\hat{\epsilon}_{i,\tau}/h_n)x_i^{(A)}x_i^{\top}}{nh},
\end{equation}
\begin{equation}
\label{Xi0_RB}
\Xi_{0,h}^{RB}=\sum_{i=1}^{n}\frac{\phi(\hat{\epsilon}_{i,\tau}^{A^c}/h)x_i^{(A)}x_i^{(A^c)\top}}{nh},
\end{equation}
and 
\begin{equation}
\label{Psi2}
    \Psi_{0,m}=\sum_{i=1}^{m^*}\frac{1}{\surd{(mm^*)}}(\varpi_{i,m}^{A^c}-\frac{m}{n}\varpi_{1,n}^{A^c})V_i,
\end{equation}
where $\varpi_{i,m}^{A^c}=\sum_{j=i}^{i+m-1}\psi_{\tau}(\hat{\epsilon}_{j,\tau}^{A^c})x_j^{(A^c)}$ and $\hat{\epsilon}_{j,\tau}^{A^c}=y_j-x_j^{(A^c)\top}\hat{\beta}_n^{A^c}$.

The algorithm for implementing $T_n^{LR}$ and $T_n^{RB}$ is given below .
\begin{enumerate}[Step 1:]
    \item Select the block size $m$ with the minimum volatility method and the bandwidth $h$ with cross-validation.
    \item Fit the constrained quantile regression with $x_i$ and $x_i^{(A^c)}$ as covariates, respectively, to get $\hat{\beta}_n$ and $\hat{\beta}_n^{A^c}$. 
    \begin{enumerate}
        \item For $T_n^{LR}$:  Calculate $T_n^{LR}$ as Equation (\ref{LR}), and get $\Xi_{h}$ and $\Xi_{0,h}$ as Equation (\ref{Xi1}) and (\ref{Xi0}).
        \item For $T_n^{RB}$: Calculate $T_n^{RB}$ as Equation (\ref{RB}), and get $\Xi_{h}^{RB}$ and $\Xi_{0,h}^{RB}$ as Equation (\ref{Xi_RB}) and (\ref{Xi0_RB}).
    \end{enumerate}
    \item For $i=1,\ldots,m^*$, generate i.i.d standard normal variables $V_i$, and  calculate $\Psi_m$ and $\Psi_{0,m}$ as Equation (\ref{Psi1}) and (\ref{Psi2}).  Compute $\hat{\Lambda}_n=\mathcal{P}_{Q,\Xi_{h}}(n^{\frac{1}{4}}\hat{\beta}_n+\hat{\Upsilon}_n)-n^\frac{1}{4}\hat{\beta}_n$ and $\hat{\Lambda}_{0,n}=\mathcal{P}_{Q,\Xi_{0,h}}(n^{\frac{1}{4}}\hat{\beta}_n^{A^c}+\hat{\Upsilon}_{0,n})-n^\frac{1}{4}\hat{\beta}_n^{A^c}$ where $\hat{\Upsilon}_n=\Xi_{h}^{-1}\Psi_m$ and  and $\hat{\Upsilon}_{0,n}=\Xi_{0,h}^{-1}\Psi_{0,m}$.
    \begin{enumerate}
        \begin{sloppypar}
        \item For $T_n^{LR}$: Compute the bootstrapped test statistic $T_n^{LR*}=g_1(\hat{\Lambda}_{0,n},\Xi_{0,h},\Psi_{0,m})-g_1(\hat{\Lambda}_{n},\Xi_{h},\Psi_m)$.
        \item For $T_n^{RB}$: Compute the bootstrapped test statistic $T_n^{RB*}=g_2((\Xi_{h}^{RB})^{-1}\hat{\Lambda}_{n},(\Xi_{0,h}^{RB})^{-1}\hat{\Lambda}_{0,n},D_n)$.
        \end{sloppypar}
    \end{enumerate}

    \item Repeat Step 3 for $B$ iterations to get $\{T_{n1}^{LR*},\ldots,T_{nB}^{LR*}\}$ and $\{T_{n1}^{RB*},\ldots,T_{nB}^{RB*}\}$.
    The resulting $p$-values for  $T_n^{LR}$ and  $T_n^{RB}$ are $B^{-1}  \sum_b I(T_n^{LR}>T_{nb}^{LR*})$ and $B^{-1}  \sum_b I(T_n^{RB}>T_{nb}^{RB*})$, respectively. 
 
\end{enumerate}

Theorem 2 shows the projected multiplier bootstrap procedure is consistent, and can detect local alternatives with $n^{-1/2}$ rate.

\begin{sloppypar}
\begin{theorem}
Suppose that regularity conditions (C1)-(C5) hold, the block size $m$ satisfies $m \rightarrow \infty$, $m/n\rightarrow 0$, and the bandwidth $h$ satisfies $h \log^2 n \rightarrow 0$, $nh^3\log^{-2} n \rightarrow \infty$, we have

(\romannumeral1) $\hat{\Lambda}_{n}\Rightarrow\Theta_{Q,\mathcal{M}}(\beta_0,\mathcal{M}^{-1}U)$;

(\romannumeral2) For any $\alpha \in (0,1)$, let $d_{\alpha}^{LR}$ be the $(1-\alpha)$-th quantile of $g_1(\hat{\Lambda}_{0,n},\Xi_{0,h},\Psi_{0,m})-g_1(\hat{\Lambda}_{n},\Xi_{h},\Psi_m)$ conditional on the original data set, then under $H_0$, $\pr(T_n^{LR}>d_{\alpha}^{LR})\rightarrow\alpha$ as $n\rightarrow\infty$; 

(\romannumeral3) For any $\alpha \in (0,1)$, let $d_{\alpha}^{RB}$ be the $(1-\alpha)$-th quantile of $g_2((\Xi_{h}^{RB})^{-1}\hat{\Lambda}_{n},(\Xi_{0,h}^{RB})^{-1}\hat{\Lambda}_{0,n},D_n)$ conditional on the original data set, then under $H_0$, $\pr(T_n^{RB}>d_{\alpha}^{RB})\rightarrow\alpha$ as $n\rightarrow\infty$;

(\romannumeral4) Under $H_{\alpha}$: $\beta_0^{(A)}=L_n\in Q$ with $n^{-1/2}|L_n|\rightarrow \infty$, $\tilde{L}_n(n h^3)^{-1/2}\rightarrow0$, $\tilde{L}_nn^2/h^3\rightarrow0$  and $m \log^8 n/n \rightarrow 0$, where $\tilde{L}_n=\max(|L_n|,n^{-1/2} \log^2 n)$, we have $\pr(T_n^{LR}>d_{\alpha}^{LR})\rightarrow 1$ as $n\rightarrow\infty$ and $\pr(T_n^{RB}>d_{\alpha}^{RB})\rightarrow 1$ as $n\rightarrow\infty$.
\end{theorem}
\end{sloppypar}
\section{Simulations}

In this section we conduct Monte Carlo simulation to examine the performance of our proposed method. 

Our simulation is based on the model
\begin{equation}
\label{sim1}
    y_i=\beta_0+\beta_1x_i+\epsilon_i,
\end{equation}
with the inequality constraints $\beta_0\geq 0$ and $\beta_1\geq 0$. We set $\beta_0=1$ throughout the simulation and vary $\beta_1$.

We consider the following three different settings based on model (\ref{sim1}):
\begin{enumerate}[(i)]
    \item Generate $\{x_i\}$ and $\{\epsilon_i\}$ from two independent \textsc{ar}$(1)$ models with coefficient 0.5. This setting represents a stationary model with homoscedastic errors. 
    \item Generate $\{x_i\}$ from an \textsc{ar}$(1)$ model with a coefficient of $-0.3$ when $t\leq0.5$ and an \textsc{ar}$(1)$ model with a coefficient of 0.3 when $t>0.5$. Furthermore, let $e_i=0.7\cos{(2\pi t)}e_{i-1}+\eta_i$ where ${\eta_i}$ follows an i.i.d standard normal distribution independent of $\{x_i\}$. Let $\epsilon_i=e_i-F_{e_i}^{-1}(\tau)$. This setting represents a non-stationary time series with homoscedastic errors.
    \item Generate $\{x_i\}$ as in Setting (\romannumeral2), let $\epsilon_i=(1+x_i^2)^{1/2}(e_i-F_{e_i}^{-1}(\tau))/2$, where $e_i$ is generated as in Setting (\romannumeral2). This is a non-stationary quantile regression model with heteroscedastic errors, as $\{\epsilon_i\}$ and $\{x_i\}$ are dependent. 
\end{enumerate}

Table \ref{alpha} summarizes the empirical coverage probability of the confidence interval of $\beta_1$ in the binding case ($\beta_1=0$) and non-binding case ($\beta_1=1$). We also present the type 1 of error of testing  
\begin{equation}
    H_0:\beta_1=0 \quad vs \quad H_{\alpha}:\beta_1 > 0
    \label{hy}
\end{equation}
with $T_n^{LR}$ and $T_n^{RB}$. A normal kernel function is used to get $\Xi_{h}$. The block size $m$ is chosen with the minimum volatility method and bandwidth $h$ is chosen with cross validition. The bootstrap size is 1000 and the simulation uses 1000 generated data sets.

According to Table \ref{alpha}, the performance of our method is reasonably good across all three settings. The accuracy of $T_n^{LR}$ and $T_n^{RB}$ does not deteriorate as the complexity of model increases, and the coverage probability of the confidence interval is similar under the binding and non-binding cases.

\begin{table}
    \centering 
    \begin{tabular}{c cc cc cc cc}
         \hline
         \hline
         \multicolumn{1}{c}{\textbf{}}&\multicolumn{4}{c}{$\alpha=5\%$}&\multicolumn{4}{c}{$\alpha=10\%$}\\
         \cmidrule(lr){2-5}\cmidrule(lr){6-9}
         \multicolumn{1}{c}{}&\multicolumn{2}{c}{$n=400$}&\multicolumn{2}{c}{$n=800$}&\multicolumn{2}{c}{$n=400$}&\multicolumn{2}{c}{$n=800$}\\
         \cmidrule(lr){2-3}\cmidrule(lr){4-5}\cmidrule(lr){6-7}\cmidrule(lr){8-9}
         Quantile&0.5&0.8&0.5&0.8&0.5&0.8&0.5&0.8\\
         \hline
         \textbf{Setting (\romannumeral1)} & & & & & & & &\\
         \cmidrule(lr){1-1}
         RB&5.4 &6.1 &4.7 &6.3 &10.4 &9.8 &9.9 &11.0 \\

         LR&5.3 &4.9 &5.3 &5.5 &10.2 &8.3 &8.8 &10.4 \\
         
         CI(B)&95.6 &96.0 &95.8 &95.4 &91.9 &93.1 &92.9 &92.6 \\
         
         CI(NB)&93.6 &92.6 &94.5 &93.0 &88.2 &86.5 &89.9 &87.3 \\
         \hline
         \textbf{Setting (\romannumeral2)} & & & & & & & &\\
         \cmidrule(lr){1-1}
         RB&5.6 &5.8 &5.0 &5.2 &11.1 &11.3 &8.8 &10.7 \\

         LR&5.3 &4.4 &4.7 &5.2 &8.4 &8.7 &8.6 &9.4 \\
         
         CI(B)&95.2 &95.5 &96.6 &95.7 &92.6 &92.2 &93.6 &93.3 \\
         
         CI(NB)&94.6 &93.7 &94.9 &93.8 &89.9 &88.0 &89.8 &88.3 \\
         \hline
         \textbf{Setting (\romannumeral3)} & & & & & & & &\\
         \cmidrule(lr){1-1}
         RB&6.0 &5.7 &5.5 &5.5 &10.9 &10.9 &11.2 &10.3 \\

         LR&5.6 &6.2 &5.0 &6.7 &8.9 &10.3 &9.4 &10.5 \\
         
         CI(B)&95.2 &95.5 &95.6 &95.8 &93.3 &92.8 &92.0 &92.5 \\
         
         CI(NB)&94.5 &93.2 &94.6 &93.7 &88.6 &88.8 & 89.7&88.6 \\
         \hline
         \hline
         
    \end{tabular}
    \caption{Simulated Type 1 error rates of the likelihood ratio test (LR) and rank-based test (RB), and coverage probability of the confidence interval (CI) for $\beta_1$; (B) stands for the binding case while (NB) stands for the non-binding case.}
    \label{alpha}
\end{table}

We also compared the empirical power when testing (\ref{hy}) using the Wald test (implemented by constructing the confidence interval), the likelihood ratio test and the rank-based test, with and without considering the inequality constraints. To ensure a fair comparison, we deliberately specify the block size $m$ so that the type I error rate of all the methods are very close to the nominal level 0.05 under $H_0$.
Figure \ref{power} summarizes our results under Setting (\romannumeral3), the results under the other two settings tell a similar story.

We observe in Figure \ref{power} that the methods incorporating the inequality constraints have higher power compared to those that do not. This indicates potential benefit of utilizing the information provided by the inequality constraints in achieving higher statistical power. Among the three methods that account for the inequality constraints, we notice a slight advantage of the rank-based test over the other two methods in terms of power, particularly when $\tau=0.8$.

\begin{figure}
    \centering
    \includegraphics[width=1\linewidth]{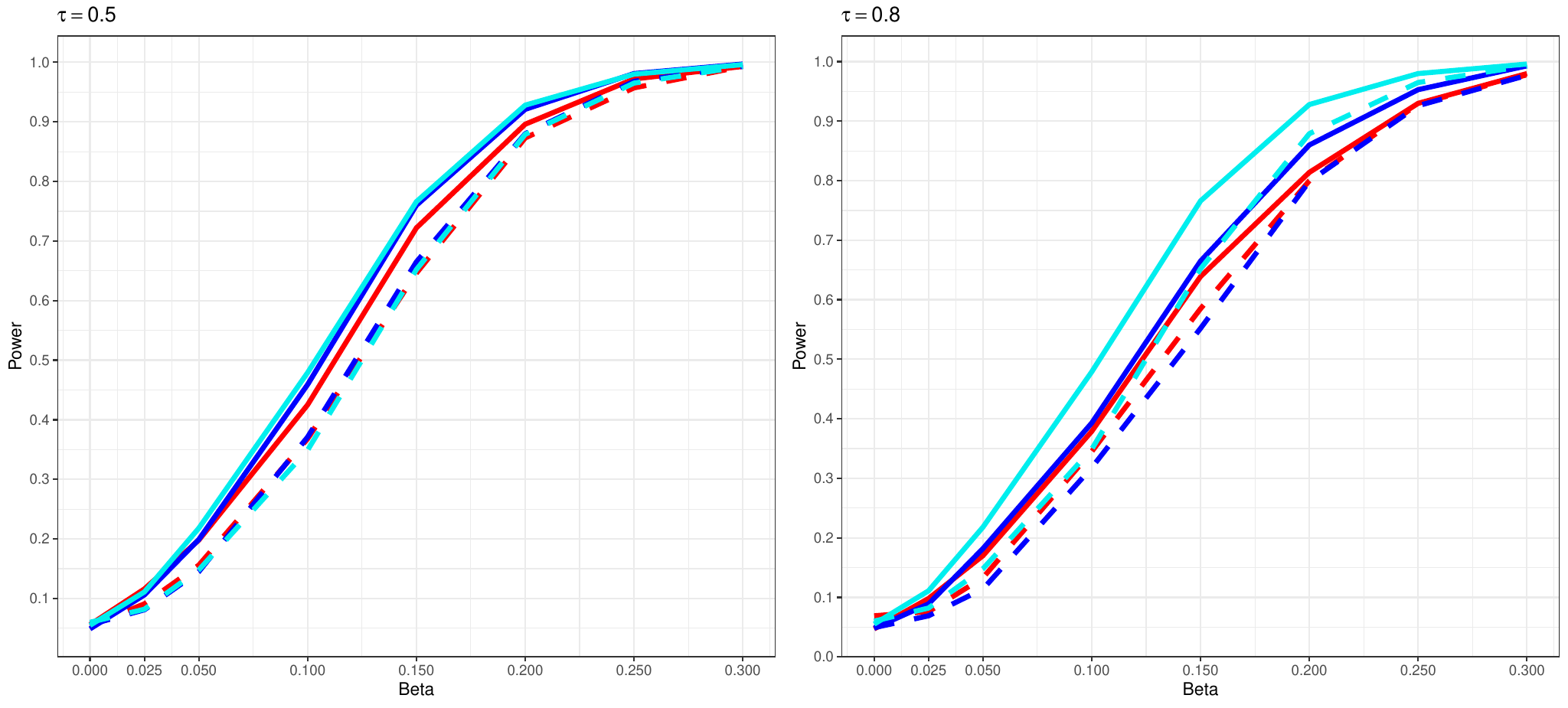}
    \caption{Simulated power of the likelihood ratio test with/without inequality constraints (solid/dashed blue), the rank-based test with/without inequality constraints (solid/dashed cyan), the Wald test with/without inequality constraints (solid/dashed red) under Setting (\romannumeral3). }
    \label{power}
\end{figure}

\section{Electricity Demand Data}
In this section, we apply our proposed method to study whether electricity demand tends to increase or decrease on school days. We use the electricity demand data in Victoria, Australia, collected between 1 January 2015 and 6 October 2020, which is available for download at https://www.kaggle.com/datasets/aramacus/electricity-demand-in-victoria-australia?resource=download. 
This data set includes the following variables: total daily electricity demand ($D_t$), electricity price ($P_t$), temperature ($T_t$), total daily sunlight energy ($S_t$), and an indicator for school days ($I_t$). Because the electricity usage pattern may change drastically in 2020 due to the Covid-19 pandemic, we focus on the 1826 days between 1 January 2015 and 31 December 2019 to ensure a consistent analysis. We consider the following model,
\begin{equation}
    \log(D_t) = \beta_0+\beta_1 \log(P_t) + \beta_2 T_t+\beta_3 T_t^2 +\beta_4 S_t+\beta_5 S_t^2+\beta_6 I_t+\epsilon_t.
\end{equation}
We include the second-order term of $T_t$ and $S_t$ in the model to account for the possible nonlinear impact of the weather condition on electricity demand.

We assume that $\beta_1\leq0$, because by the law of demand in microeconomics, `conditional on all else being equal, as the price of a good increases,  quantity demanded will decrease' (\cite{NS2012}). And we are interested in testing if $\beta_6=0$ under this constraint with our proposed tests at different quantile levels.

\begin{table}
    \centering
    \begin{tabular}{c c c c c c}
         \hline
         \hline
           & $100\hat{\beta}_6$ & LR & LR(NC) & RB & RB(NC)\\
         \hline
          $\tau=0.1$&3.33 &0.010 & 0.122 &0.027 &0.040\\
         \hline
         $\tau=0.5$&3.38 &$<0.001$ & 0.007 &$<0.001$ &$<0.001$\\
         \hline
         $\tau=0.9$&1.73 &0.030 & 0.031 &0.041 &0.037\\
         \hline
         \hline
    \end{tabular}
    \normalsize
    \caption{P-values of the proposed likelihood ratio test (LR) and the rank-based test (RB) at different quantile levels $\tau$; (NC) represents the corresponding tests ignoring the inequality constraint.}
    \label{rd1}
\end{table}


We observe in Table \ref{rd1} that, under the constraint $\beta_1\leq0$, both the likelihood ratio test and the rank-based test reject the null hypothesis at the three quantile levels considered. Our results suggest that electricity demand in Victoria tends to be higher on school days, not only at the median but also at the upper and lower tails of the distribution. We also notice that the p-values usually become larger when the inequality constraint is ignored. Specifically, the likelihood ratio test actually fails to reject the null hypothesis at $\tau=0.1$ if the constraint is ignored. This example further illustrates the benefits of considering the inequality constraint in terms of power.

\section{Regularity Conditions}
In the following regularity conditions, $\chi \in (0,1)$ and $M<\infty$ are constants that may take different values from line to line.

\begin{enumerate}[({C}1)]
    \item The process $\{x_i\}_{i=1}^n$ satisfies $$\max_{0\leq r\leq R}\sup_{b_r<t_1<t_2<b_{r+1}}\|\frac{H_r(t_1,\mathcal{F}_{-1},\mathcal{G}_0)-H_r(t_2,\mathcal{F}_{-1},\mathcal{G}_0)}{t_1-t_2}\|_4 \leq M,$$
    
    $ \Delta_4(H,k)=O(\chi^{|k|})$ and $\max_{1\leq i\leq n}\|x_i\|_{4+\eta} \leq M$, where $\eta$ is a small positive constant.
    
    \item The process $\{\epsilon_{i,\tau}\}_{i=1}^n$  satisfies $$\max_{0\leq r\leq R}\sup_{b_r<t_1<t_2<b_{r+1}}\|\frac{D_r(t_1,\mathcal{F}_0,\mathcal{G}_0)-D_r(t_2,\mathcal{F}_0,\mathcal{G}_0)}{t_1-t_2}\|_4 \leq M.$$
    
    
    Define $F_r^{(q)}(t,x|\mathcal{F}_{k-1},\mathcal{G}_k)=\frac{\partial^q}{\partial x^q} \pr(D_r(t,\mathcal{F}_k,\mathcal{G}_k)\leq x\mid\mathcal{F}_{k-1},\mathcal{G}_k)$, $b_r<t\leq b_{r+1}$. 
    There exist $c>0$ s.t. for $0\leq q\leq p$, 
    \begin{equation*}
        \begin{aligned}
            \max_{0\leq r\leq R}\sup_{b_r<t<b_{r+1},|u|< c}&\|F_r^{(q)}(t,H_r(t,\mathcal{F}_{k-1},\mathcal{G}_{k})^{\top}u\mid\mathcal{F}_{k-1},\mathcal{G}_k)\\
            &-F_r^{(q)}(t,H_r(t,\mathcal{F}_{k-1}^*,\mathcal{G}_{k}^*)^{\top}u\mid\mathcal{F}_{k-1}^*,\mathcal{G}_k^*)\|_4 = O(\chi^k)
        \end{aligned}
    \end{equation*}

    \item For some $c>0$ and $\epsilon>0$, assume that 
    $$\epsilon \leq \min_{0\leq r\leq R}\inf_{b_r<t<b_{r+1},|x|<c}|f_r(t,x\mid\mathcal{F}_{-1},\mathcal{G}_0)|\leq \max_{0\leq r\leq R}\sup_{b_r<t<b_{r+1},x\in \mathbb{R}}|f_r(t,x\mid\mathcal{F}_{-1},\mathcal{G}_0)|\leq M,$$

    $$\max_{0\leq r\leq R}\sup_{b_r<t<b_{r+1},-c<x_1<x_2<c}|\frac{f_r(t,x_1\mid\mathcal{F}_{-1},\mathcal{G}_0)-f_r(t,x_2\mid\mathcal{F}_{-1},\mathcal{G}_0)}{x_1-x_2}|\leq M,$$
    
    $$\max_{0\leq r\leq R}\sup_{b_r<t_1<t_2<b_{r+1}}\|\frac{f_r(t_1,0\mid\mathcal{F}_{-1},\mathcal{G}_0)-f_r(t_2,0\mid\mathcal{F}_{-1},\mathcal{G}_0)}{t_1-t_2}\|_2 \leq M.$$

    \item Let $\lambda_1(\cdot)$ denotes the smallest eigenvalue of a matrix, assume 

$$\min_{0\leq r\leq R}\inf_{b_r<t<b_{r+1}} \lambda_1(E(H_r(t,\mathcal{F}_{-1},\mathcal{G}_0)H_r(t,\mathcal{F}_{-1},\mathcal{G}_0)^{\top}))\geq \epsilon$$
    for some $\epsilon>0$.
    

    \item Assume $\inf_{0\leq t\leq 1}\lambda_1(\Omega(t))\geq \epsilon$
    for some $\epsilon>0$.

\end{enumerate}

Conditions (C1) and (C2) require the data generation mechanism of $\{x_i\}_{i=1}^n$ and $\{\epsilon_{i,\tau}\}_{i=1}^n$ to be smooth between breakpoints, and the processes $\{x_i\}_{i=1}^n$ and $\{F_r^{(q)}(t,x\mid\mathcal{F}_{k-1},\mathcal{G}_k)\}_{i=1}^n$ to be short-term dependent with exponentially decreasing dependence measure.
 Conditions (C1)-(C2) together imply that the process $\{x_i\psi_{\tau}(\epsilon_{i,\tau})\}_{i=1}^n$ is piecewise stochastic Lipschiz continuous and short-term dependent, and therefore $n^{-1/2}\sum_i x_i\psi_{\tau}(\epsilon_{i,\tau})$ convergences to a Gaussian process by Proposition 5 of \cite{Z2013}. Condition (C3) assumes that the conditional density $f_r(t,x\mid\mathcal{F}_{-1},\mathcal{G}_0)$ is bounded away from 0 and infinity, Lipschiz continuous in $x$ and stochastic Lipschiz continuous in $t$. These are common assumptions required to establish the asymptotic properties of quantile regression. Condition (C4) and (C5) require the design matrix and the long-term covariance matrix $\Omega(t)$ to be positive definite, and imply that the limit of  $K_n/n$ is positive definite.

\section*{Acknowledgement}
We thank Anran Jia for her contribution to an earlier version of this paper.

\appendix

\section*{Appendix}

\begin{proof}[of Lemma 1]

Without loss of generality, assume $\beta_0=0$. Let
$$M_n(\beta)=\sum_{i=1}^n\big(\psi_{\tau}(\epsilon_{i,\tau}-x_i^{\top}\beta)x_i-
E(\psi_{\tau}(\epsilon_{i,\tau}-x_i^{\top}\beta)x_i|\mathcal{F}_{i-1},\mathcal{G}_i)\big),$$
and
$$N_n(\beta)=\sum_{i=1}^n\big(E(\psi_{\tau}(\epsilon_{i,\tau}-x_i^{\top}\beta)x_i|\mathcal{F}_{i-1},\mathcal{G}_i)-
E(\psi_{\tau}(\epsilon_{i,\tau}-x_i^{\top}\beta)x_i)\big).$$

Denote $|v|=\surd({v^T v})$ for a vector $v$. According to Lemma A.1 and A.2 of the supplemental material of \cite{WZ2018}, for any sequence $\delta_n\rightarrow0$,
\begin{equation}
\begin{aligned}
    &\sup\limits_{|\beta|\leq\delta_n}|M_n(\beta)+N_n(\beta)-(M_n(0)+N_n(0))|\\
    =&\sup\limits_{|\beta|\leq\delta_n}|\sum_{i=1}^n\psi_{\tau}(\epsilon_{i,\tau}-x_i^{\top}\beta)x_i-
\sum_{i=1}^nE(\psi_{\tau}(\epsilon_{i,\tau}-x_i^{\top}\beta)x_i)-
\sum_{i=1}^n\psi_{\tau}(\epsilon_{i,\tau})x_i|\\
    =&O_p\big((\sum_{i=1}^n\nu_i(\delta_n))^{1/2}\log n+\surd{n}\delta_n\big),
\end{aligned}
\label{thm1.0}
\end{equation}
where $\nu_i(\delta_n)=E\big(|x_i|^2 I(|\epsilon_{i,\tau}|<|x_i|\delta_n)\big)$.

Let $\delta_n=n^{-1/2} \log n$ and integrate the above equation w.r.t $\beta$, we have
\begin{equation}
\begin{aligned}
    &\sup\limits_{|\beta|\leq\delta_n}|\int_0^{\beta}\sum_{i=1}^n\psi_{\tau}(\epsilon_{i,\tau}-x_i^{\top}s)x_i-
\sum_{i=1}^nE\big(\psi_{\tau}(\epsilon_{i,\tau}-x_i^{\top}s)x_i\big)-
\sum_{i=1}^n\psi_{\tau}(\epsilon_{i,\tau})x_ids|\\
=&\sum_{i=1}^n\big(\rho_{\tau}(\epsilon_{i,\tau})-\rho_{\tau}(\epsilon_{i,\tau}-x_i^{\top}\beta)\big)+\sum_{i=1}^n E\big(\rho_{\tau}(\epsilon_{i,\tau}-x_i^{\top}\beta)-\rho_{\tau}(\epsilon_{i,\tau})\big)-\beta^{\top}G_n\\
=&o_p(1).
\end{aligned}
\label{thm1}
\end{equation}

By the convexity of $\rho_{\tau}$ and Taylor expansion, for any $k=1,\ldots,n$,
\begin{equation}
\begin{aligned}
\label{thm1.3}
    E\big(\rho_{\tau}(\epsilon_{i,\tau}-\frac{kx_i^{\top}\beta}{n})-&\rho_{\tau}(\epsilon_{i,\tau}-\frac{(k-1)x_i^{\top}\beta}{n})\big) \leq
    -E\big(\frac{x_i^{\top}\beta}{n}\psi_{\tau}(\epsilon_{i,\tau}-\frac{kx_i^{\top}\beta}{n})\big)\\
    &\leq-E\{\frac{x_i^{\top}\beta}{n}\big(\tau-F_i(0)-\frac{kx_i^{\top}\beta}{n}f_i(0)-\frac{kx_i^{\top}\beta}{n}\int_0^1 f_i(\frac{kx_i^{\top}\beta}{n}s)-f_i(0) ds\big)\}\\
    &\leq E\big(\frac{k(x_i^{\top}\beta)^2}{n^2}f_i(0)\big)+O_p(\frac{k^2(x_i^{\top}\beta)^3}{n^3}),
\end{aligned}
\end{equation}
where we write $F_i(x) = F_r(\frac{i}{n},x\mid\mathcal{F}_{i-1},\mathcal{G}_i)$ and $f_i(x)= f_r(\frac{i}{n},x\mid\mathcal{F}_{i-1},\mathcal{G}_i)$ for simplicity.

Similarly, we have
\begin{equation}
\label{thm1.4}
    E\big(\rho_{\tau}(\epsilon_{i,\tau}-\frac{kx_i^{\top}\beta}{n})-\rho_{\tau}(\epsilon_{i,\tau}-\frac{(k-1)x_i^{\top}\beta}{n})\big) \geq E\big(\frac{(k-1)(x_i^{\top}\beta)^2}{n^2}f_i(0)\big)+O_p(\frac{(k-1)^2(x_i^{\top}\beta)^3}{n^3}).
\end{equation}

Summing Equation (\ref{thm1.3}) and (\ref{thm1.4}) over $k$ from 1 to $n$, we get
\begin{equation*}
    \frac{n-1}{2n} \beta^{\top} E(f_i(0)x_i x_i^{\top}) \beta  \leq E\{\rho_{\tau}(\epsilon_{i,\tau}-x_i^{\top}\beta)-\rho_{\tau}(\epsilon_{i,\tau})\}+ O_p((x_i^{\top}\beta)^3) \leq
    \frac{n+1}{2n} \beta^{\top} E(f_i(0)x_i x_i^{\top}) \beta.
\end{equation*}
Namely,
\begin{equation}
\label{thm1.5}
    \sum_{i=1}^nE\{\rho_{\tau}(\epsilon_{i,\tau}-x_i^{\top}\beta)-\rho_{\tau}(\epsilon_{i,\tau})\}=\frac{1}{2}\beta^{\top}K_n\beta+o_p(1).
\end{equation}

We have the desired result by inserting Equation (\ref{thm1.5}) into Equation (\ref{thm1}).
\end{proof}

\begin{proposition}
Under regularity conditions (C1)-(C5), we have

(\romannumeral1)  $|\tilde{\beta}_n-\beta_0|= o_p(n^{-1/2}\log n)$ and $|\hat{\beta}_n-\beta_0|= o_p(n^{-1/2}\log n)$;

(\romannumeral2) $(\tilde{\beta}_n-\beta_0)-K_n^{-1} G_n=o_p(n^{-1/2})$;

(\romannumeral3) $|\hat{\beta}_n - \hat{\gamma}_n| = o_p(n^{-1/2})$.
\end{proposition}

\begin{proof}[of Proposition A1]

Without loss of generality, assume $\beta_0=0$.

(\romannumeral1) For any $c>0$, let $\delta_n=(n^{-1/2}\log n)c$ and $\lambda_n$ be the smallest eigenvalue of $K_n/n$. $\lambda_n$ is strictly positive by (A3) and (A4), and  $\beta^{\top} K_n \beta \geq_p (\log n)^2 c^2 \lambda_n$ for $|\beta|=\delta_n $, 
By Proposition 3.1 of \cite{WZ2018}, we also have $\beta^{\top} G_n = O_p(c\log n )$ for $|\beta|=\delta_n $. Note that \cite{WZ2018} considered general M-estimation and required the stochastic Lipschiz continous and short-term dependent conditions in (C1) and (C2) to hold for higher moments, we checked that their conditions could be relaxed in the context of quantile regression.

By Lemma 1,
$$\pr\{\inf_{|\beta|= \delta_n} \sum_{i=1}^n(\rho_{\tau}(\epsilon_{i,\tau}-x_i^{\top}\beta)-\rho_{\tau}(\epsilon_{i,\tau})) \leq 0\}\rightarrow 0,$$
and by the convexity of $\rho_{\tau}$,
$$\pr\{\inf_{|\beta|\geq \delta_n} \sum_{i=1}^n(\rho_{\tau}(\epsilon_{i,\tau}-x_i^{\top}\beta)-\rho_{\tau}(\epsilon_{i,\tau})) \leq 0\}\rightarrow 0.$$

Namely, for any $c>0$, we have
$$ \pr\{|n^{1/2}(\log n)^{-1}\tilde{\beta}_n|>c\} \rightarrow 0 $$

The desired results for $\hat{\beta}_n$ can be derived in the same way.

(\romannumeral2) Let $\bar{\beta} = K_n^{-1} G_n=O_p(n^{-1/2})$. By Lemma 1,

\begin{equation}
\label{pro1.1}
    \sum_{i=1}^n\{\rho_{\tau}(\epsilon_{i,\tau}-x_i^{\top}\bar{\beta})-\rho_{\tau}(\epsilon_{i,\tau})\}+\frac{1}{2} \bar{\beta}^{\top} K_n \bar{\beta}=o_p(1).
\end{equation}

Let $\delta_n = n^{-1/2} c$, by Lemma 1 and Equation (\ref{pro1.1}),
$$\sup\limits_{|\beta-\bar{\beta}|=\delta_n}|\sum_{i=1}^n\big(\rho_{\tau}(\epsilon_{i,\tau}-x_i^{\top}\beta)-\rho_{\tau}(\epsilon_{i,\tau}-x_i^{\top}\bar{\beta})\big)-\frac{1}{2}(\beta-\bar{\beta})^{\top}K_n(\beta-\bar{\beta})|=o_p(1). $$

Because $(\beta-\bar{\beta})^{\top}K_n(\beta-\bar{\beta}) \geq_p \lambda_n c^2$ when $|\beta-\bar{\beta}|=\delta_n$, 

$$\pr\{\inf_{|\beta-\bar{\beta}|\geq \delta_n} \sum_{i=1}^n(\rho_{\tau}(\epsilon_{i,\tau}-x_i^{\top}\beta)-\rho_{\tau}(\epsilon_{i,\tau}-x_i^{\top}\bar{\beta})) \leq 0\}\rightarrow 0.$$
Namely $\surd{n}(\tilde{\beta}-\bar{\beta})=o_p(1)$.

(\romannumeral3)
By Lemma 1 and Proposition A1(\romannumeral2), we have
\begin{equation}
    \sup\limits_{|\beta|\leq n^{-1/2}\log n}|\sum_{i=1}^n(\rho_{\tau}(\epsilon_{i,\tau}-x_i^{\top}\beta)-\rho_{\tau}(\epsilon_{i,\tau})+\beta^{\top}K_n\tilde{\beta}-\frac{1}{2}\beta^{\top}K_n\beta|=o_p(1).
\end{equation}

Recall that $$\hat{\beta}_n = \argmin\limits_{\beta\in Q}\sum_{i=1}^n (\rho_{\tau}(\epsilon_{i,\tau}-x_i^{\top} \beta)-\rho_{\tau}(\epsilon_{i,\tau})),$$
and 
\begin{equation*}
\begin{aligned}
   \hat{\gamma}_n&=\argmin\limits_{\beta\in Q}(\beta-\tilde{\beta}_n)^{\top}K_n(\beta-\tilde{\beta}_n)\\
    &=\argmin\limits_{\beta\in Q} (-\beta^{\top}K_n\tilde{\beta}+\frac{1}{2}\beta^{\top}K_n\beta).
\end{aligned}
\end{equation*}
Let $\delta_n = n^{-1/2} c$, by elementary calculation,
$$\sup\limits_{|\beta-\hat{\gamma}_n|=\delta_n,\beta\in Q}|\sum_{i=1}^n(\rho_{\tau}(\epsilon_{i,\tau}-x_i^{\top}\beta)-\rho_{\tau}(\epsilon_{i,\tau}-x_i^{\top}\hat{\gamma}_n))-\frac{1}{2}(\beta-\hat{\gamma}_n)^{\top}K_n(\beta-\hat{\gamma}_n)-\bigtriangledown g(\hat{\gamma}_n)^{\top}(\beta-\hat{\gamma}_n)|=o_p(1),$$
where
$g(\beta)=-\beta^{\top}K_n\tilde{\beta}+\frac{1}{2}\beta^{\top}K_n\beta$. Because $g(\cdot)$ is a convex function and $\hat{\gamma}_n$ is the minimum of $g(\cdot)$, $\bigtriangledown g(\hat{\gamma}_n)^{\top}(\beta-\hat{\gamma}_n)$ is non-negative.
Also notice that $(\beta-\hat{\gamma}_n)^{\top}K_n(\beta-\hat{\gamma}_n)\geq \lambda_n c^2$ when $|\beta-\hat{\gamma}_n|=\delta_n$, we have
$$\pr\{\inf_{|\beta-\hat{\gamma}_n|\geq \delta_n,\beta\in Q} \sum_{i=1}^n(\rho_{\tau}(\epsilon_{i,\tau}-x_i^{\top}\beta)-\rho_{\tau}(\epsilon_{i,\tau}-x_i^{\top}\hat{\gamma}_n)) \leq 0\}\rightarrow 0.$$
We have the desired result.
\end{proof}

\begin{proof}[of Theorem 1]

(\romannumeral1) We have $n^{-1/2} G_n\Rightarrow U$ and $|K_n/n-\mathcal{M}|=o_p(1)$ by Proposition 3.1 and Theorem 3.1 of \cite{WZ2018}, respectively. Then $ \surd{n}(\tilde{\beta}_n-\beta_0)\Rightarrow\mathcal{M}^{-1}U$ by the continuous mapping theorem and Proposition A1.

By Proposition A1 and Proposition 1(\romannumeral1)(\romannumeral4) of \cite{Z2015}
\begin{equation}
\begin{aligned}
    \surd{n}(\hat{\beta}_n-\beta_0) &= \surd{n} \mathcal{P}_{Q,\frac{k_n}{n}}(\tilde{\beta}_n) - \surd{n} \beta_0 + o_p(1)\\
    &= \mathcal{P}_{Q,\frac{k_n}{n}}(\surd{n}\beta_0 + \surd{n}K_n^{-1}G_n) - \surd{n} \beta_0 + o_p(1).
\end{aligned}
\end{equation}
Proposition 1(\romannumeral2)(\romannumeral4)(\romannumeral5) of \cite{Z2015} then gives
$\surd{n}(\hat{\beta}_n-\beta_0)\Rightarrow\Theta_{Q,\mathcal{M}}(\beta_0,\mathcal{M}^{-1}U)$.

(\romannumeral2) By Lemma 1 and Proposition A1,  $$\sum_{i=1}^n(\rho_{\tau}(y_i-x_i^{\top}\hat{\beta}_n)-\rho_{\tau}(y_i-x_i^{\top}\beta_0))=\frac{1}{2}(\hat{\beta}_n-\beta_0)^{\top}K_n(\hat{\beta}_n-\beta_0)-G_n(\hat{\beta}_n-\beta_0)+o_p(1).$$

By Theorem 1(\romannumeral1),  $$\frac{1}{2}(\hat{\beta}_n-\beta_0)^{\top}K_n(\hat{\beta}_n-\beta_0)-G_n(\hat{\beta}_n-\beta_0)\Rightarrow g_1(\Theta_{Q,\mathcal{M}}(\beta_0,\mathcal{M}^{-1}U),\mathcal{M},U).$$

Then the desired result follows trivially.

(\romannumeral3) By Equation (\ref{thm1.0}) and Proposition A1,
$$\frac{1}{\surd{n}}\sum_{i=1}^n\psi_{\tau}(y_i-x_i^{\top}\hat{\beta})x_i^{(A)}=\frac{1}{\surd{n}}\sum_{i=1}^n\psi_{\tau}(\epsilon_{i,\tau})x_i^{(A)}+\frac{1}{n}\sum_{i=1}^n E(f_r(\frac{i}{n},0\mid\mathcal{F}_{i-1},\mathcal{G}_i)x_i^{(A)} x_i^{\top})\surd{n}(\hat{\beta}_n-\beta_0)+o_p(1).$$
Therefore similar to the above arguments, we have
$S_{1,n}\Rightarrow U^{(A)}+\mathcal{M}^{RB}\Theta_{Q,\mathcal{M}}\big(\beta_0,\mathcal{M}^{-1}U)$ and $S_{0,n}\Rightarrow U^{(A)}+\mathcal{M}_0^{RB}\Theta_{Q_0,\mathcal{M}_0}(\beta_0^{(A^c)},\mathcal{M}_0^{-1}U^{(A^c)})$, which leads to the desired result.
\end{proof}

\begin{proof}[of Theorem 2]
    
(\romannumeral1)
By Theorem 3.3 and 3.4 in \cite{WZ2018}, $\hat{\Upsilon}_{n}\Rightarrow\mathcal{M}^{-1}U$.

By Proposition A1, $n^{\frac{1}{4}}\hat{\beta}_n-n^{\frac{1}{4}}\beta_0=o_p(1)$, then we have $\hat{\Lambda}_{n}\Rightarrow\Theta_{Q,\mathcal{M}}(\beta_0,\mathcal{M}^{-1}U)$.

(\romannumeral2) and (\romannumeral3) are obvious.

(\romannumeral4)
By Proposition B.1 of the supplementary material of \cite{WZ2018}, under $H_{\alpha}$,
$$\hat{\Upsilon}_{0,n}= \mathcal{M}_0^{-1}U^{(A^c)} + O_p(\surd{m}\tilde{L}_n+(n\sum_{i=1}^n\nu_i(\tilde{L}_n))^{1/2}\log n).$$ Thus by Proposition 1 of \cite{Z2015}, the fastest rate at which $d_{\alpha}^{LR}$ converges to infinity is $m \tilde{L}_n^2 \log^2 n$. However, similar to the proof of Lemma 1, under $H_{\alpha}$, $T_n^{LR}$ go to infinity at rate $n |L_n|^2$, which is faster than $d_{\alpha}^{LR}$. Therefore, $\pr(T_n^{LR}>d_{\alpha}^{LR})\rightarrow1$. Similarly, $\pr(T_n^{LR}>d_{\alpha}^{RS})\rightarrow1$.
\end{proof}

\bibliography{ref}

@article{BRW1992,
 ISSN = {10170405, 19968507},
 author = {Z. D. Bai and C. Radhakrishna Rao and Y. Wu},
 journal = {Statistica Sinica},
 number = {1},
 pages = {237--254},
 publisher = {Institute of Statistical Science, Academia Sinica},
 title = {M-ESTIMATION OF MULTIVARIATE LINEAR REGRESSION PARAMETERS UNDER A CONVEX DISCREPANCY FUNCTION},
 urldate = {2024-01-11},
 volume = {2},
 year = {1992}
}

@book{K2005,
    author    = "Roger Koenker",
    title     = "Quantile regression",
    year      = "2005",
    publisher = "Cambridge University Press",
    address   = " "
}

@book{PRW1999,
    author    = "Dimitris N. Politis and Joseph P. Romano and Michael Wolf",
    title     = "Subsampling",
    year      = "1999",
    publisher = "Springer",
    address   = " "
}

@book{NS2012,
    author    = "Walter Nicholson and Christopher Snyder",
    title     = "Microeconomic Theory: Basic Principles and Extensions",
    year      = "2012",
    publisher = "South-Western College Pub",
    address   = " "
}

@article{HN1999,
  author =       "Xuming He and Pin Ng",
  title =        "COBS: qualitatively constrained smoothing  via linear programming",
  journal =      "Computational Statistics",
  volume =       "14",
  number =       "",
  pages =        "315-337",
  year =         "1999",
  DOI =          " "
}

@article{KN2005,
  author =       "Roger Koenker and Pin Ng",
  title =        "Inequality constrained quantile regression",
  journal =      "Sankhya: The Indian Journal of Statistics",
  volume =       "67",
  number =       "2",
  pages =        "418-440",
  year =         "2005",
  DOI =          " "
}

@article{Z2013,
  author =       "Zhou Zhou",
  title =        "Heteroscedasticity and Autocorrelation Robust Structural Change Detection",
  journal =      "Journal of the American Statistical Association",
  volume =       "108",
  number =       "502",
  pages =        "726-740",
  year =         "2013",
  DOI =          " "
}

@article{Z2015,
 ISSN = {13697412, 14679868},
 author = {Zhou Zhou},
 journal = {Journal of the Royal Statistical Society. Series B (Statistical Methodology)},
 number = {2},
 pages = {349--371},
 publisher = {[Royal Statistical Society, Wiley]},
 title = {Inference for non-stationary time series regression with or without inequality constraints},
 urldate = {2024-01-04},
 volume = {77},
 year = {2015}
}

@article{WZ2018,
  author =       "Weichi Wu and Zhou Zhou",
  title =        "Gradient-based structral change detection for nonstationary time series M-estimation",
  journal =      "Annals of Statistics",
  volume =       "46",
  number =       "3",
  pages =        "1197-1224",
  year =         "2018",
  DOI =          " "
}

@article{KB1978,
  author =       "Roger Koenker and Gilbert Bassett",
  title =        "Regression Quantiles",
  journal =      "Econometrica",
  volume =       "46",
  number =       "1",
  pages =        "33--50",
  year =         "1978",
  DOI =          " "
}

@article{RS2019,
  author =       "Yeonwoo Rho and Xiaofeng Shao",
  title =        "Bootstrap-assisted unit root testing with piecewise locally stationary errors",
  journal =      "Econometric Theory",
  volume =       "35",
  number =       "1",
  pages =        "142-166",
  year =         "2019",
  DOI =          " "
}

@article{KM1999,
  author =       "Roger Koenker and José A. F. Machado",
  title =        "Goodness of Fit and Related Inference Processes for Quantile Regression",
  journal =      "Journal of the American Statistical Association",
  volume =       "94",
  number =       "448",
  pages =        "1296-1310",
  year =         "1999",
  DOI =          " "
}

@article{KZ1996,
 ISSN = {02664666, 14694360},
 author = {Roger Koenker and Quanshui Zhao},
 journal = {Econometric Theory},
 number = {5},
 pages = {793--813},
 publisher = {Cambridge University Press},
 title = {Conditional Quantile Estimation and Inference for {ARCH} Models},
 urldate = {2024-01-24},
 volume = {12},
 year = {1996}
}

@article{KX2006,
  author =       "Roger Koenker and Zhijie Xiao",
  title =        "Quantile Autoregression",
  journal =      "Journal of the American Statistical Association",
  volume =       "101",
  number =       "475",
  pages =        "980-1006",
  year =         "2006",
  DOI =          " "
}

@article{P1991,
title = {Asymptotic behavior of regression quantiles in non-stationary, dependent cases},
journal = {Journal of Multivariate Analysis},
volume = {38},
number = {1},
pages = {100-113},
year = {1991},
issn = {0047-259X},
doi = {https://doi.org/10.1016/0047-259X(91)90034-Y},
author = {Stephen Portnoy}
}

@article{KM1994,
title = {Regression Quantiles and Related Processes Under Long Range Dependent Errors},
journal = {Journal of Multivariate Analysis},
volume = {51},
number = {2},
pages = {318-337},
year = {1994},
issn = {0047-259X},
doi = {https://doi.org/10.1006/jmva.1994.1065},
author = {H.L. Koul and K. Mukherjee},
}

@article{P2019,
title = {Asymptotic inference for the constrained quantile regression process},
journal = {Journal of Econometrics},
volume = {213},
number = {1},
pages = {174-189},
year = {2019},
note = {Annals: In Honor of Roger Koenker},
issn = {0304-4076},
doi = {https://doi.org/10.1016/j.jeconom.2019.04.010},
author = {Thomas Parker}
}

@article{DRW2019,
author = {Rainer Dahlhaus and Stefan Richter and Wei Biao Wu},
title = {{Towards a general theory for nonlinear locally stationary processes}},
volume = {25},
journal = {Bernoulli},
number = {2},
publisher = {Bernoulli Society for Mathematical Statistics and Probability},
pages = {1013 -- 1044},
year = {2019},
doi = {10.3150/17-BEJ1011}
}

@article{KP2015,
 ISSN = {13697412, 14679868},
 author = {Jens-Peter Kreiss and Efstathios Paparoditis},
 journal = {Journal of the Royal Statistical Society. Series B (Statistical Methodology)},
 number = {1},
 pages = {267--290},
 publisher = {[Royal Statistical Society, Wiley]},
 title = {Bootstrapping locally stationary processes},
 urldate = {2024-01-04},
 volume = {77},
 year = {2015}
}

@article{DPV2011,
 ISSN = {01621459},
 author = {Holger Dette and Philip Preuss and Mathias Vetter},
 journal = {Journal of the American Statistical Association},
 number = {495},
 pages = {1113--1124},
 publisher = {[American Statistical Association, Taylor & Francis, Ltd.]},
 title = {A Measure of Stationarity in Locally Stationary Processes With Applications to Testing},
 urldate = {2024-01-04},
 volume = {106},
 year = {2011}
}

@article{DP2021,
author = {Srinjoy Das and Dimitris N. Politis},
title = {Predictive Inference for Locally Stationary Time Series With an Application to Climate Data},
journal = {Journal of the American Statistical Association},
volume = {116},
number = {534},
pages = {919-934},
year = {2021},
publisher = {Taylor & Francis},
doi = {10.1080/01621459.2019.1708368}

}

@article{HHY2019,
author = {Lixia Hu, Tao Huang and Jinhong You},
title = {Estimation and Identification of a Varying-Coefficient Additive Model for Locally Stationary Processes},
journal = {Journal of the American Statistical Association},
volume = {114},
number = {527},
pages = {1191-1204},
year = {2019},
publisher = {Taylor & Francis},
doi = {10.1080/01621459.2018.1482753}
}

@article{K2022,
author = {Daisuke Kurisu},
title = {{Nonparametric regression for locally stationary functional time series}},
volume = {16},
journal = {Electronic Journal of Statistics},
number = {2},
publisher = {Institute of Mathematical Statistics and Bernoulli Society},
pages = {3973 -- 3995},
year = {2022},
doi = {10.1214/22-EJS2041}
}

@article{BR2023,
author = {Sumanta Basu and Suhasini Subba Rao},
title = {{Graphical models for nonstationary time series}},
volume = {51},
journal = {The Annals of Statistics},
number = {4},
publisher = {Institute of Mathematical Statistics},
pages = {1453 -- 1483},
year = {2023},
doi = {10.1214/22-AOS2205}
}

@article{LZL2020,
title = {Generalized $l_1$-penalized quantile regression with linear constraints},
journal = {Computational Statistics \& Data Analysis},
volume = {142},
pages = {106819},
year = {2020},
issn = {0167-9473},
doi = {https://doi.org/10.1016/j.csda.2019.106819},

author = {Yongxin Liu and Peng Zeng and Lu Lin}
}

@article{QY2015,
title = {Nonparametric estimation and inference on conditional quantile processes},
journal = {Journal of Econometrics},
volume = {185},
number = {1},
pages = {1-19},
year = {2015},
issn = {0304-4076},
doi = {https://doi.org/10.1016/j.jeconom.2014.10.008},

author = {Zhongjun Qu and Jungmo Yoon}
}

@article{WLY2021,
title = {Penalized and constrained {LAD} estimation in fixed and high dimension},
journal = {Statistical Papers},
volume = {63},
number = { },
pages = {53-95},
year = {2021},
issn = {},
doi = {https://doi.org/10.1016/j.jeconom.2014.10.008},
author = {Xiaofei Wu and Rongmei Liang and Hu Yang}
}

@article{RD2019,
author = {Stefan Richter and Rainer Dahlhaus},
title = {{Cross validation for locally stationary processes}},
volume = {47},
journal = {The Annals of Statistics},
number = {4},
publisher = {Institute of Mathematical Statistics},
pages = {2145 -- 2173},
year = {2019},
doi = {10.1214/18-AOS1743}

}

@incollection{Powell1991,
  author    = {James L. Powell},
  title     = {ESTIMATION OF MONOTONIC REGRESSION MODELS UNDER QUANTILE RESTRICTIONS},
  booktitle = {Nonparametric and Semiparametric Methods in Econometrics and Statistics},
  publisher = {Cambridge University Press},
  year      = {1991},
  editor    = {William A. Barnett and James Powell and George E. Tauchen},
  address   = {Cambridge},

}

\end{document}